\documentclass[11pt,a4paper,showkeys]{article}

\usepackage{amsmath}
\usepackage{amssymb}
\usepackage{amsthm}
\usepackage{latexsym}
\usepackage[dvipdfmx]{graphicx}
\usepackage[dvips]{color}

\theoremstyle{plain}
\newtheorem{theorem}{Theorem}[section]
\newtheorem{corollary}[theorem]{Corollary}

\newtheorem{lemma}[theorem]{Lemma}

\theoremstyle{definition}
\newtheorem{remark}[theorem]{Remark}

\newtheorem{definition}[theorem]{Definition}

\renewcommand{\labelenumi}{\theenumi}
\renewcommand{\labelenumii}{\theenumii}
\renewcommand{\theenumi}{{\rm (\roman{enumi})}}
\renewcommand{\theenumii}{{\rm (\alph{enumii})}}

\newcommand{\bm}[1]{{\mbox{\boldmath $#1$}}}

\newcommand{\ZZ}{{\mathbb{Z}}}
\newcommand{\RR}{{\mathbb{R}}}

\newcommand{\remove}[1]{}

\usepackage{fullpage}

\usepackage{newtxtext,newtxmath}

\title{Pure Nash Equilibria in Weighted Congestion Games with 
Complementarities and Beyond\thanks{An extended abstract of this paper \cite{Tak24} 
appears in the Proceedings of the 23rd International
Conference on Autonomous Agents and Multiagent Systems (AAMAS 2024).}}

\author{
Kenjiro Takazawa\thanks{Department of Industrial and Systems Engineering, Faculty of Science and Engineering, 
Hosei University, Tokyo 184-8584, Japan.  
\texttt{takazawa@hosei.ac.jp}.  
Partially supported by 
JSPS KAKENHI Grant Number 
JP20K11699, 
Japan.
}}

\date{January 2024}

\begin{document}

\maketitle

\begin{abstract}
Congestion games offer a primary model in the study of pure Nash equilibria in non-cooperative games, and a number of generalized models have been proposed in the literature. 
One line of generalization includes weighted congestion games, in which the cost of a resource is a function of the total weight of the players choosing that resource. 
Another line includes congestion games with mixed costs, in which the cost imposed on a player is a convex combination of the total cost and the maximum cost of the resources in her strategy. 
This model is further generalized to that of congestion games with complementarities. 
For the above models, the existence of a pure Nash equilibrium is proved under some assumptions, including that the strategy space of each player is the base family of a matroid and that the cost functions have a certain kind of monotonicity. 
In this paper, we deal with common generalizations of these two lines, namely weighted matroid congestion games with complementarities, and its further generalization. 
Our main technical contribution is a proof of the existence of pure Nash equilibria in these generalized models under a simplified assumption on the monotonicity, which provide a common extension of the previous results. 
We also present some extensions on the existence of pure Nash equilibria in player-specific and weighted matroid congestion games with mixed costs. 
\begin{description}
    \item[Keywords] Non-Cooperative Game Theory; 
Pure Nash Equilibrium; 
Congestion Game; 
Matroid
\end{description}
\end{abstract}

\renewcommand{\labelenumi}{\theenumi}
\renewcommand{\labelenumii}{\theenumii}
\renewcommand{\theenumi}{{\rm (\roman{enumi})}}
\renewcommand{\theenumii}{{\rm (\alph{enumii})}}

\newcommand{\agg}{g}
\newcommand{\lex}{\preceq_{\mathrm{lex}}}
\newcommand{\succlex}{\succeq_{\mathrm{lex}}}
\newcommand{\lexstrict}{\prec_{\mathrm{lex}}}
\newcommand{\lexeq}{\sim_{\mathrm{lex}}}
\newcommand{\lexneqq}{\prec_{\mathrm{lex}}}

\section{Introduction}
\label{SECintro}

The analysis of the existence of pure Nash equilibria in non-cooperative games is a central topic in algorithmic game theory \cite{NRTV07,Rou16}, 
and the model of \emph{congestion games} \cite{BV23book} play an important role in this topic. 
A congestion game models the human behavior of selecting the best strategy 
in a situation where congestion results in a larger cost. 
In a congestion game, 
each strategy is a set of \emph{resources}. 
The cost of a resource is represented by a function of the number of the players using it. 
The entire cost imposed on a player is the sum of the costs of the resources in her strategy.  
For a formal description of congestion games, 
see Section \ref{SECclassical}.

The most important feature of congestion games is that 
every congestion game is a \emph{potential game}, 
which implies that every congestion game has a pure Nash equilibrium \cite{Ros73a}. 
Further, 
the converse is also true: 
every potential game is represented by a congestion game \cite{MS96}. 
Since those primary results, 
a number of generalized models of congestion games have been proposed. 
As is mentioned below, 
however, 
a pure Nash equilibrium may not exist in those generalized models, 
and its existence is proved under certain assumptions. 
The purpose of this paper is to provide a model 
commonly generalizing previous models, 
while preserving the existence of a pure Nash equilibrium. 

\subsection{Previous Work}
\label{SECprevious}

In this section, 
we describe two major directions of 
generalizations of congestion games. 

\subsubsection{Weighted and Player-Specific Congestion Games}

In a main direction of the study of congestion games, the cost functions of the resources are generalized. 
One model in this direction is a \emph{weighted congestion game}: 
each player has a nonnegative weight; 
and 
the cost of a resource 
is a function of the total weight of the players using it. 
Another model is  a \emph{player-specific congestion game}, 
in which 
the resources have a cost function specific to each player. 

Both of the weighted and player-specific congestion games 
do not have 
a pure Nash equilibrium in general, 
but 
are guaranteed to possess one under certain assumptions. 
A congestion game is referred to as a \emph{singleton game} 
if every strategy of every player consists of a single resource, 
and as a \emph{matroid game}
if the strategy space of each player is a base family of a matroid; 
see Section \ref{SECmatroid} for the definition of matroids. 
Milchtaich \cite{Mil96} proved that every singleton player-specific congestion game has a pure Nash equilibrium. 
Ackermann, R\"oglin, and V\"ocking \cite{ARV09} extended this result to 
matroid congestion games, 
i.e.,\ 
every player-specific matroid congestion game has a pure Nash equilibrium. 
A further generalized model on polymatroids is studied by Harks, Klimm, and Peis \cite{HKP18}.  
Ackermann, R\"oglin, and V\"ocking \cite{ARV09} also proved that 
every weighted matroid congestion game have a pure Nash equilibrium, 
and this result is extended to a further generalized model of a \emph{matroid congestion game with set-functional costs} by Takazawa \cite{Tak19}. 
Some of these results are explained in detail in Section \ref{SECwp}.

\subsubsection{Congestion Games with Bottleneck Costs, Mixed Costs, and Complementarities}

In another major direction of generalization, 
the cost imposed on a player is not simply defined as the sum of the costs of the resources in the strategy. 
A typical model is a \emph{bottleneck congestion game}, 
in which 
the cost on a player is defined as the maximum cost of the resources in her strategy. 
Banner and Orda \cite{BO07} proved the existence of pure Nash equilibria 
and analyzed the price of anarchy 
in 
certain kinds of network bottleneck congestion games, 
in which a strategy of a player is a flow in a given network.

Harks, Klimm, and M\"ohring \cite{HKM13} introduced a generalized class of \emph{congestion games with lexicographical improvement property}. 
They proved 
that every bottleneck congestion game belongs to this class, 
and that 
every congestion game in this class possesses a \emph{strong equilibrium} \cite{Aum59}, 
a stronger concept of a pure Nash equilibrium. 
Harks, Hoefer, Klimm, and Skopalik \cite{HHKS13} analyzed the complexity of finding a strong equilibrium 
in bottleneck congestion games. 
In particular, 
they presented polynomial-time algorithms for 
several classes of bottleneck congestion games, 
including 
matroid bottleneck congestion games. 

Feldotto, Leder, and Skopalik \cite{FLS18} 
proposed a common generalization of classical congestion games
and bottleneck congestion games, 
called \emph{congestion games with mixed costs}. 
In this model, 
the cost imposed on a player is a convex combination of the 
total cost and the maximum cost of the resources in her strategy; 
see Section \ref{SECmixed} for detailed description. 
For this model,  
a matroid game may fail to possess a pure Nash equilibrium. 
Feldotto, Leder, and Skopalik \cite{FLS18} 
proved the existence of 
pure Nash equilibria in certain classes of this model, 
including 
singleton games, 
and matroid games with some more property 
(Theorem \ref{THMmixed} below).

Feldotto, Leder, Skopalik \cite{FLS17} introduced a further generalized model, 
\emph{congestion games with complementarities}. 
In this model, 
the cost imposed on a player 
is determined by an \emph{aggregation function}, 
which a general function of the costs of the resources in the strategy. 
A typical example of an aggregation function is 
the $L^p$ norm of the costs of the resources. 
A detailed description of this model appears in Section \ref{SECmixed}. 
For this model, 
Feldotto, Leder, Skopalik \cite{FLS17} 
proved the existence of a pure Nash equilibrium 
in matroid games under an assumption on the monotonicity of the aggregation function, 
which is defined in a sophisticated manner (Theorem \ref{THMcomple} below). 
Related work dealing with aggregation cost functions also appears in \cite{Kuk07,Kuk15}.

\subsection{Our Contribution}

The aim of this paper is to offer a model of congestion games which commonly generalizes the previous models in the two directions mentioned in Section \ref{SECprevious}. 
In particular, 
we provide a generalized model preserving the existence of pure Nash equilibria in matroid games.

Our main contribution is a common extension of the existence of pure Nash equilibria in 
weighted congestion games and congestion games with complementarities. 
Namely, 
our model is 
\emph{weighted congestion games with complementarities},  
and 
its further generalization of 
\emph{congestion games with set-functional costs and complementarities}. 
On the basis of 
the previous results \cite{ARV09,FLS17,Tak19}, 
we prove that every game in this model has a pure Nash equilibrium  
if it is a matroid game and the aggregation function has a certain monotonicity (see Theorem \ref{THMPNE} below). 
This class includes 
weighted matroid bottleneck congestion games
and 
weighted congestion games with $L^p$-aggregation functions. 
Thus, 
our result implies the guaranteed existence of pure Nash equilibria in these models (Corollaries \ref{CORbottleneck} and \ref{CORlp} below). 
We remark here that 
the definition of 
the monotonicity of the aggregation functions in our model is 
simpler than that in the previous work \cite{FLS17}, 
which provides a simpler proof for the existence of pure Nash equilibria and 
a clearer comprehension of the model.

We also generalize the model of congestion games with mixed costs 
to those with player-specific cost functions 
and 
set-functional cost functions. 
For these models, 
we extend the result of Feldotto, Leder, and Skopalik \cite{FLS18} 
to prove the existence of pure Nash equilibria under the same assumptions. 

\subsection{Organization of the Paper}

The rest of the paper is organized as follows.  
In Section \ref{SECpre}, 
we give a detailed description of the previous models. 
In Section \ref{SECmain}, 
we present our main contribution: 
the generalized model of congestion games with set-functional costs and complementarities; 
a simplified definition of the monotonicity of the aggregation functions; and 
a proof of the existence of pure Nash equilibria in matroid games with a monotone aggregation function. 
Finally, 
in Section \ref{SECmixed-g}, 
we deal with generalizations of congestion games with mixed costs. 

\section{Preliminaries}
\label{SECpre}

In this section, 
we present some basic properties of matroids, 
followed by descriptions of classical and recent generalized models of 
congestion games. 

\subsection{Matroids}
\label{SECmatroid}

Let $E$ be a finite set and $\mathcal{S} \subseteq 2^E$ its subset family. 
The pair $(E, \mathcal{S})$ is a \emph{matroid} 
if $\mathcal{S} \neq \emptyset$ and 
it satisfies the following exchange axiom: 
\begin{description}
\item[($\mbox{EX}$)] 
for $S, S' \in \mathcal{S}$ 
and for $e\in S \setminus S'$, 
there exists $f \in S' \setminus S$ 
such that 
$(S \setminus \{e\}) \cup \{f\} \in \mathcal{S}$. 
\end{description}
The subset family $\mathcal{S}$ is called the \emph{base family}, 
and a set $S\subseteq E$ belonging to $\mathcal{S}$ is called a \emph{base} 
of the matroid $(E,\mathcal{S})$. 
It is noteworthy that all bases have the same cardinality, 
i.e.,\ 
$|S| = |S'|$ for each $S,S'\in \mathcal{S}$ for every matroid $(E,\mathcal{S})$.

Optimization over matroids 
has formed an extensive research field in combinatorial optimization, 
due to its nice structure 
and wide range of application 
 (see, e.g.,\ 
\cite{Mur00,Rec89,Sch03}). 
Below we show some properties of matroids, 
which will be used in subsequent sections. 

First, the one-sided exchange axiom ($\mbox{EX}$) is equivalent to 
the following stronger 
two-sided exchange property. 
\begin{lemma}[\cite{Bru69}]
\label{LEMsimul}
Every matroid $(E, \mathcal{S})$ has the following simultaneous exchange property: 
\begin{description}
\item[($\mbox{EX}\pm$)] 
for $S, S' \in \mathcal{S}$ 
and for $e\in S \setminus S'$, 
there exists $f \in S' \setminus S$ 
such that 
$(S \setminus \{e\}) \cup \{f\} \in \mathcal{S}$ 
and 
$(S' \setminus \{f\}) \cup \{e\} \in \mathcal{S}$.  
\end{description}
\end{lemma}

The next lemma is derived from the exchange axiom (EX) and will be used in the analysis of congestion games with mixed costs. 

\begin{lemma}
\label{LEMgreedy}
Let $(E, \mathcal{S})$ be a matroid 
and let $c_e\in \RR$ denote the cost of an element $e\in E$. 
Then, 
a base $S \in \mathcal{S}$ minimizing 
the total cost 
$\sum_{e\in S}c_e$ 
among all bases in $\mathcal{S}$ 
also minimizes 
the maximum cost 
$\max_{e\in S}\{c_e\}$
among all bases in $\mathcal{S}$. 
\end{lemma}

\begin{proof}
Let a base $S_1 \in \mathcal{S}$ minimize $\sum_{e\in S}c_e$ among all bases. 
Suppose to the contrary that $S_1$ does not minimize $\max_{e\in S}\{c_e\}$, 
i.e.,\ 
there exists another base $S_2\in \mathcal{S}$ such that 
an element $e^*\in S_1\setminus S_2$ satisfies 
\begin{align*}
c_{e^*} > c_f \quad \mbox{for each element $f\in S_2$}. 
\end{align*}
Then, 
it follows from the exchange axiom ($\mbox{EX}$) that 
there exists an element $f^*\in S_2\setminus S_1$ 
such that 
$S_1' = (S_1 \setminus \{e^*\})\cup \{f\}\in \mathcal{S}$. 
We now obtain 
$$
\sum_{e\in S_1'}c_e = \sum_{e\in S_1}c_e - c_{e^*} +  c_{f^*} < \sum_{e\in S_1}c_e, 
$$
contradicting that $S_1$ is a base minimizing $\sum_{e\in S}c_e$. 
\end{proof}

\subsection{Classical Model of Congestion Games}
\label{SECclassical}

Let $\ZZ_+$ and $\RR_+$ denote the sets of the nonnegative integers and the nonnegative real numbers, respectively. 
The classical model of congestion games is represented as 
\begin{align*}
(N, E, (\mathcal{S}_i)_{i \in N}, (c_e)_{e\in E} ). 
\end{align*}
Here, 
$N = \{1,\ldots, n\}$ and $E$ denote the sets of the players and the resources, 
respectively. 
For each player $i \in N$, 
the resource subset family $\mathcal{S}_i \subseteq 2^E$ denotes the strategy space of $i$. 
For each resource $e \in E$, 
a cost function $c_e \colon \ZZ_+ \to \RR_+$ 
is associated. 

In a congestion game 
$(N, E, (\mathcal{S}_i)_{i \in N}, (c_e)_{e\in E} )$, 
the cost imposed on a player is determined in the following way. 
A \emph{strategy profile} $\mathbf{S}=(S_1,\ldots, S_n)$ is a collection of 
the strategies of all players, 
i.e.,\ 
$S_i\in \mathcal{S}_i$ for each $i\in N$. 
For a strategy profile $\mathbf{S}=(S_1,\ldots, S_n)$ and 
a resource $e \in E$,  
let
$N_e(\mathbf{S})$ 
denote 
the set of the players choosing $e$ in her strategy, 
and 
$n_e(\mathbf{S})$ its cardinality. 
Namely, 
$$N_e(\mathbf{S}) = \{ i \in N \mid e \in S_i \}, 
\quad 
n_e(\mathbf{S})  = \left|\{ i \in N \mid e \in S_i \}\right|. 
$$
Now $c_e(n_e(\mathbf{S})) \in \RR_+$ represents 
the cost of using a resource $e$ in a strategy profile $\mathbf{S}$. 
The cost $\gamma_i(\mathbf{S})$ imposed on a player $i \in N$ 
in a strategy profile $\mathbf{S}$ 
is defined as 
the total cost of the resources in her strategy $S_i$, 
that is, 
\begin{align*}
\gamma_i(\mathbf{S})= \sum_{e \in S_i}c_e(n_e(\mathbf{S})). 
\end{align*}

For 
a strategy profile $\mathbf{S}=(S_1,\ldots, S_n)$ 
and 
a player $i\in N$,  
let $\mathbf{S}_{-i}$ denote 
the collection of the strategies in $\mathbf{S}$ other than $S_i$, 
i.e., 
$\mathbf{S}_{-i}=(S_1,\ldots, S_{i-1},S_{i+1},\ldots, S_n)$. 
A \emph{pure Nash equilibrium} is a strategy profile in which 
no player has an incentive to change her strategy. 
Namely, 
a strategy profile $\mathbf{S}=(S_1,\ldots, S_n)$ is a pure Nash equilibrium 
if 
\begin{align*}
\gamma_i(\mathbf{S}) \le \gamma_i(\mathbf{S}_{-i},S_i') 
\end{align*}
holds for each player $i\in N$ and 
each strategy $S_i'\in \mathcal{S}_i$. 
As mentioned in Section \ref{SECintro}, 
every congestion game has a pure Nash equilibrium \cite{Ros73a}.

\subsection{Weighted and Player-Specific Congestion Games}
\label{SECwp}

The model of \emph{weighted congestion games} is described as 
\begin{align*}
(N, E, (\mathcal{S}_i)_{i \in N}, (c_e)_{e\in E}, (w_i)_{i \in N} ). 
\end{align*}
The difference from the classical model is the presence of the nonnegative weight 
$w_i \in \RR_+$ of each player $i\in N$. 
For each resource $e\in E$, 
the argument of the cost function $c_e$ is 
not the cardinality of the players using $e$, 
but the sum of the weights of those players. 
For a subset $N' \subseteq N$ of players, 
let $w(N')$ denote the total weight of the players in $N'$, 
i.e.,\ 
$w(N')=\sum_{i\in N'}w_i$. 
In a strategy profile $\mathbf{S}=(S_1,\ldots, S_n)$, 
the cost $\gamma_i(\mathbf{S})$ imposed on a player $i\in N$ is described as 
\begin{align*}
\gamma_i(\mathbf{S}) = \sum_{e \in S_i}c_e( w(N_e(\mathbf{S})) ). 
\end{align*}
Hereafter, 
we assume that $c_e$ is monotonically nondecreasing 
for each resource $e\in E$, 
i.e., 
\begin{align*}
\mbox{$x < y $ implies $c_e(x)\le c_e(y)$} \quad (x,y\in \ZZ_+).
\end{align*}

In contrast to the classical model, 
weighted congestion games do not necessarily have pure Nash equilibria. 
Ackermann, R\"{o}glin, and V\"{o}cking \cite{ARV09} proved that 
a weighted congestion game has a pure Nash equilibrium 
if it is a \emph{matroid game}. 
Recall that a congestion game is a  matroid game if 
the strategy space of each player is a base family of a matroid. 
\begin{theorem}[\cite{ARV09}]
\label{THMweighted}
Every weighted matroid congestion game possesses a pure Nash equilibrium. 
\end{theorem}
Ackermann, R\"{o}glin, and V\"{o}cking \cite{ARV09} also proved that 
the matroid property is 
the maximal assumption 
in a certain sense 
ensuring the existence of a pure Nash equilibrium.

A further generalized model of weighted congestion games is proposed by Takazawa \cite{Tak19}. 
A distinctive feature of this model is that 
the cost function $c_e$ of each resource $e\in E$ is a set function 
defined on the power set $2^N$ of the players: 
$c_e \colon 2^N \to \RR_+$. 
Similarly to the cost functions in weighted congestion games, 
this cost function 
$c_e$ 
is assumed to be monotonically nondecreasing, 
that is, 
$$
\mbox{$X \subseteq Y$ implies $c_e (X) \le c_e (Y)$ \quad ($X,Y\subseteq N$)}.
$$
In a strategy profile $\mathbf{S}=(S_1,\ldots, S_n)$, 
the cost $\gamma_i(\mathbf{S})$ imposed on a player $i\in N$ is described as 
\begin{align*}
\gamma_i(\mathbf{S})=\sum_{e \in S_i}c_e(N_e(\mathbf{S})). 
\end{align*}
We refer to this model as a \emph{congestion game with set-functional costs}.

It is straightforward to see that this model generalizes weighted congestion games. 
To see this fact, 
for an instance of a weighted congestion game $(N, E, (\mathcal{S}_i)_{i \in N}, (c_e)_{e\in E}, (w_i)_{i \in N} )$, 
it suffices to define a set function $c'_e \colon 2^N \to \RR_+$ by 
\begin{align}
\label{EQreduction}
c'_e(N') = c_e (w(N')) \quad (N' \subseteq N)     
\end{align}
for each resource $e\in E$. 
Takazawa \cite{Tak19} extended Theorem \ref{THMweighted} 
as follows. 
\begin{theorem}[\cite{Tak19}]
\label{THMmonotone}
Every matroid congestion game with set-functional costs has a pure Nash equilibrium. 
\end{theorem}

\begin{remark}
The same kind of generalization with set-functional costs is proposed by Kiyosue and Takazawa \cite{KT22} for \emph{budget games} \cite{DFRS15,DFRS19}. 
A budget game is a variant of congestion game in which the players share the budgets, or utilities, of the resources. 
In the original model \cite{DFRS15,DFRS19}, 
the utility of a resource is simply divided by the number of the players using that resource and allocated to those players. 
In the generalized model \cite{KT22}, 
the allocated utility is represented by a monotonically nondecreasing set function. 
Similarly to the situation of congestion games described in Theorems \ref{THMweighted} and \ref{THMmonotone}, 
the fact that the strategy space of each player is a base family of a matroid 
plays an essential role in guaranteeing the existence of a pure Nash equilibrium 
in both of the original and generalized models of 
budget games. 
\end{remark}

In a \emph{player-specific congestion game}, 
a cost function is defined for each pair of a player and a resource. 
Namely, 
a player-specific congestion game represented by 
\begin{align*}
(N, E, (\mathcal{S}_i)_{i \in N}, (c_{i,e})_{i\in N,e\in E} ), 
\end{align*}
where 
$c_{i, e}\colon \ZZ_+ \to \RR_+$ is a monotonically nondecreasing function for each $i\in N$ and $e\in E$. 
The cost $\gamma_i(\mathbf{S})$ imposed on a player $i\in N$ 
in a strategy profile $\mathbf{S}=(S_1,\ldots, S_n)$ is defined as 
$$
\gamma_i(\mathbf{S}) = \sum_{e\in S_i}c_{i,e}(n_e(\mathbf{S})).
$$

As mentioned in Section \ref{SECintro}, 
every player-specific singleton congestion game has a pure Nash equilibrium \cite{Mil96}. 
Ackermann, R\"oglin, and V\"ocking \cite{ARV09} extended the existence of pure Nash equilibria in player-specific matroid congestion games. 
\begin{theorem}[\cite{ARV09}]
\label{THMps}
Every player-specific matroid congestion game possesses a pure Nash equilibrium. 
\end{theorem}

\subsection{Congestion Games with Mixed Costs and with Complementarities}
\label{SECmixed}

Feldotto, Leder, and Skopalik \cite{FLS18} 
proposed a common generalization of classical congestion games
and bottleneck congestion games, 
called \emph{congestion games with mixed costs}. 
A congestion game with mixed costs is represented by a tuple 
\begin{align*}
(N, E, (\mathcal{S}_i)_{i \in N}, (\ell_e)_{e\in E}, (b_e)_{e \in E}, (\alpha_i)_{i\in N}). 
\end{align*}
For each resource $e \in E$, 
let $\ell_e \colon \ZZ_+ \to \RR_+$ and $b_e \colon \ZZ_+ \to \RR_+$ denote 
its \emph{latency cost function} and \emph{bottleneck cost function}, 
respectively. 
We assume that both 
$\ell_e$ and $b_e$ are monotonically nondecreasing. 
For each player $i \in N$, 
a real number $\alpha _i \in [0,1]$ denotes her \emph{preference value}. 
In a strategy profile $\mathbf{S}=(S_1,\ldots, S_n)$, 
the cost $\gamma_i(\mathbf{S})$ imposed on a player $i \in N$ is defined as  
\begin{align*}
\gamma_i(\mathbf{S}) = 
\alpha_i\cdot \sum_{e \in S_i}\ell_e(n_e(\mathbf{S})) + (1 - \alpha_i)\cdot\max_{e \in S_i}\{ b_e (n_e(\mathbf{S})) \}. 
\end{align*}
It is straightforward to see that 
a congestion game with mixed costs 
is exactly a classical congestion game 
if $\alpha_i=1$ for each player $i \in N$, 
and 
a bottleneck congestion game 
if $\alpha_i=0$ for each player $i \in N$.  

On the existence of pure Nash equilibria in congestion games with mixed costs, 
Feldotto, Leder, and Skopalik \cite{FLS18} proved the following theorem. 
For a congestion game with mixed costs, 
the cost functions 
$\ell_e$  and $b_e$ $(e\in E)$
have \emph{monotone dependence} if 
there exists a monotonically nondecreasing function $d \colon \RR_+ \to \RR_+$ satisfying that 
\begin{align*}
b_e (x) = d(\ell_e(x)) \qquad (x \in \ZZ_+)
\end{align*} 
for each resource $e \in E$. 

\begin{theorem}[\cite{FLS18}]
\label{THMmixed}
A congestion game 
with mixed costs 
has a pure Nash equilibrium 
in the following three cases. 
\begin{enumerate}
\item
\label{ENUsingleton}
	It is a singleton game. 
\item
\label{ENU01}
	It is a matroid game and $\alpha_i \in \{0,1\}$ for each player $i\in N$. 
\item
\label{ENUmd}
	It is a matroid game and the cost functions have monotone dependence. 
\end{enumerate}
\end{theorem}

Feldotto, Leder, and Skopalik \cite{FLS17} proposed a further generalized model of 
a \emph{congestion game with complementarities}, 
which is represented by a tuple 
\begin{align*}
(N,E, (\mathcal{S}_i)_{i\in N}, (c_e)_{e\in E}, g).
\end{align*}
In this model, 
$g\colon \RR_+^r \to \RR$ denotes the aggregation function, 
which is defined on $\RR_+^r$ for a positive integer $r$, 
and 
we assume that the sizes of the strategies of all players are equal to $r$, 
i.e.,\ 
$|S_i| = r$ for each player $i\in N$ and each strategy $S_i\in \mathcal{S}_i$.  
For a player $i\in N$ and a strategy profile $\mathbf{S}=(S_1,\ldots, S_n)$, 
where $S_i =\{e_1,\ldots, e_{r}\}$, 
the cost $\gamma_i(\mathbf{S})$ on $i$ in $\mathbf{S}$ is described as 
\begin{align*}
\gamma_i(\mathbf{S})= g(c_{e_1} ( n_{e_1}(\mathbf{S}) ), \ldots, c_{e_{r}} ( n_{e_{r}}(\mathbf{S}) ) ). 
\end{align*}

It is straightforward to see that congestion games with complementarities 
generalize classical congestion games and bottleneck congestion games. 
An important special case analyzed in detail in \cite{FLS17} is 
\emph{congestion games with $L^p$-aggregation functions}, 
i.e., 
\begin{align}
\label{EQlp}
g(x_{1}, \ldots, x_r ) 
= \left( \sum_{j=1}^r {x_j}^{p} \right)^{\frac{1}{p}}. 
\end{align}
for some value $p \ge 1$. 
It is also noteworthy that, 
if the aggregation function is player-specific, 
then this model includes 
congestion games with mixed costs.

Feldotto, Leder, Skopalik \cite{FLS17} proved the existence of a pure Nash equilibrium 
under the assumption that it is a matroid game 
and 
the aggregation function 
$g$ is \emph{weakly monotone}. 
Weak monotone functions 
generalize the monotonically nondecreasing functions,  
and 
are defined in a sophisticated manner. 
The readers are referred to their original paper \cite{FLS17} 
for the definition of weakly monotone functions.

\begin{theorem}[\cite{FLS17}]
\label{THMcomple}
Every matroid congestion game with complementarities 
has a pure Nash equilibrium 
if 
the aggregation function is 
weakly monotone. 
\end{theorem}

\section{Weighted Congestion Games with Complementarities and Beyond}
\label{SECmain}

The central subject of this paper is 
described in this section. 
We deal with 
a common generalization of weighted congestion games and congestion games with complementarities, 
referred to as \emph{weighted congestion games with complementarities}, 
and a further generalized model of \emph{congestion games with set-functional costs and complementarities}. 
Below we investigate 
the latter model, 
because the former model can be represented as the latter model with 
cost function $c_e$ defined by \eqref{EQreduction} for each resource $e\in E$.

\subsection{Our Model}

A \emph{congestion game with set-functional costs and complementarities} is described by a tuple
\begin{align*}
(N, E, (\mathcal{S}_i)_{i\in N}, (c_e)_{e\in E}, \agg). 
\end{align*}
Again we assume that 
the aggregation function $g\colon \RR_+^r \to \RR$ is defined on $\RR_+^r$ for some positive integer $r$ 
and 
$|S_i|=r$ for each player $i\in N$ and each strategy $S_i\in \mathcal{S}_i$. 
The difference from congestion games with complementarities described in Section \ref{SECmixed} is that 
the cost function $c_e\colon 2^N\to \RR_+$ is a set function for each resource $e\in E$. 
Consequently, 
for a player $i\in N$ and a strategy profile $\mathbf{S}=(S_1,\ldots, S_n)$, 
where $S_i =\{e_1,\ldots, e_{r}\}$, 
the cost $\gamma_i(\mathbf{S})$ on $i$ in $\mathbf{S}$ is described as 
\begin{align*}
\gamma_i(\mathbf{S})= g(c_{e_1}  (N_{e_1}(\mathbf{S}) ), \ldots, c_{e_{r}} ( N_{e_{r}}(\mathbf{S}) ) ). 
\end{align*}

The monotonicity  of the cost functions $c_e$ ($e\in E$) and the aggregation function $g$ 
in the previous models is represented in the following manner. 
First, 
in the aggregation of the costs of the resources in a strategy of a player, 
we take no account of the order of the resources. 
Formally, 
we assume that 
$$
g(\bm{v}) = g(\sigma(\bm{v}))
$$
for each cost vector $\bm{v} \in \RR_+^{r}$ and 
each permutation $\sigma$ of a set of $r$ elements. 
Then, 
we adopt the following definition of weakly monotone functions, 
which is simpler than that in \cite{FLS17}. 
\begin{definition}
\label{DEFwm}
Let $r$ be a positive integer. 
A function $g \colon \RR_+^r \to \RR$ is \emph{weakly monotone} if, 
for each $x,y \in \RR_+$, 
it holds that 
$$g(\bm{v}',x) \le g(\bm{v}',y) \quad \mbox{for each $\bm{v}' \in \RR^{r-1}_+$}$$
or 
$$g(\bm{v}',x) \ge g(\bm{v}',y) \quad \mbox{for each $\bm{v}' \in \RR^{r-1}_+$}.$$
\end{definition}
It is straightforward to see that 
a function $g \colon \RR_+^r \to \RR$ such that 
$$
\mbox{$\bm{u} \le \bm{v}$ implies $g(\bm{u}) \le g(\bm{v})$ \quad($\bm{u},\bm{v}\in \RR^r$)}
$$
is weakly monotone, 
where $\bm{u} \le \bm{v}$ means $u_j \le v_j$ for each $j=1,\ldots, r$. 
For instance, 
the $L^p$-aggregation function \eqref{EQlp} is weakly monotone. 

Associated with 
Definition \ref{DEFwm} 
of the weak monotonicity of the aggregation function $g$, 
the monotonicity of the cost function $c_e \colon 2^N \to \RR$ ($e\in E$) 
is adjusted in the following way. 
\begin{definition}
\label{DEFmonotonecost}
For a weakly monotone aggregation function $g\colon \RR^r_+\to \RR$, 
a cost function $c\colon 2^N\to \RR_+$ is \emph{monotonically nondecreasing} 
with respect to $g$ if, 
for each $X,Y \subseteq N$ with $X \subseteq Y$, 
it holds that 
$$
g(\bm{v}', c(X)) \le 
g(\bm{v}', c(Y)) 
\quad \mbox{for each $\bm{v}' \in \RR^{r-1}_+$}.
$$    
\end{definition}
Definition \ref{DEFmonotonecost} means that 
more players on one resource result in larger cost, 
even in the weak monotonicity of the aggregation function. 

\subsection{Pure Nash Equilibria in Matroid Games}

Our main technical contribution is the following theorem, 
which guarantees the existence of a pure Nash equilibrium 
in a matroid congestion game with set-functional costs and complementarities. 
This theorem amounts to a common extension of 
Theorems \ref{THMweighted}, \ref{THMmonotone}, and \ref{THMcomple}. 

\begin{theorem}
\label{THMPNE}
A matroid congestion game 
$(N, E, (\mathcal{S}_i)_{i\in N}, (c_e)_{e\in E}, \agg)$ 
with set-functional costs and complementarities 
has a pure
Nash equilibrium 
if the aggregation function $g\colon \RR_+^r \to \RR$ is weakly monotone 
and the cost function $c_e\colon 2^N\to \RR_+$ is monotonically nondecreasing 
with respect to $g$
for each resource $e\in E$. 
\end{theorem}

Before a proof, 
we present 
the following corollaries of Theorem \ref{THMPNE}. 
These corollaries 
are immediately derived from Theorem \ref{THMPNE}, 
and
guarantee the existence of pure Nash equilibria in some generalized models of congestion games.

\begin{corollary}
\label{CORbottleneck}
Every weighted matroid bottleneck congestion game 
in which the strategies of all players have the same cardinality 
has a pure Nash equilibrium. 
\end{corollary}

\begin{corollary}
\label{CORlp}
Every weighted matroid congestion game with $L^p$-aggregation functions 
in which the strategies of all players have the same cardinality 
has a pure Nash equilibrium. 
\end{corollary}

Our proof for Theorem \ref{THMPNE} is obtained by extending the arguments in \cite{ARV09,Tak19} and \cite{FLS17}. 
In preparation for the proof, 
we first define a total preorder $\preceq_g$ 
on $\RR_+$, 
on the basis of the weak monotonicity of the aggregation function $g$. 
Recall that a relation $\preceq$ on a set $X$ is a \emph{preorder} 
if $x\preceq x$ for each $x\in X$, and 
$x\preceq y$ and $y\preceq z$ implies $x\preceq z$ 
for each $x,y,z\in X$. 
A preorder is \emph{total} if 
$x \preceq y$ or $y \preceq x$ holds for each $x,y\in X$. 
\begin{definition}[Total preorder $\preceq_g$ on $\RR_+$]
Let $g\colon \RR_+^r \to \RR$ be a weakly monotone function. 
For two nonnegative values $x,y\in \RR_+$, 
we define $x\preceq_g y$ if 
$$\mbox{$g(\bm{v}',x)\le g(\bm{v}',y)$ for each $\bm{v}'\in \RR_+^{r-1}$}.$$ 
Further, 
we define 
$x \sim_g y$ if both $x\preceq_g y$ and $y\preceq_g x$ hold, 
and 
$x \prec_g y$ if $x\preceq_g y$ but $y\not\preceq_g x$.

\end{definition}

It directly follows from Definition \ref{DEFwm} that the relation $\preceq_g$ is a total preorder. 
By employing this total preorder $\preceq_g$ 
and the simultaneous exchange property of matroids (Lemma \ref{LEMsimul}), 
we then show that, if a strategy profile is not a pure Nash equilibrium, 
a player has an \emph{improving local move}. 
For a strategy profile $\mathbf{S}=(S_1,\ldots, S_n)$ and a player $i\in N$, 
a strategy $S_i' \in \mathcal{S}_i$ of $i$ is an \emph{improving move} in $\mathbf{S}$ if 
the cost imposed on $i$ becomes smaller by choosing $S_i'$ instead of $S_i$, 
i.e., 
$\gamma_i(\mathbf{S}_{-i},S_i') < \gamma_i(\mathbf{S})$. 
An improving move $S_i'$ is called an \emph{improving local move} 
if $|S_i' \setminus S_i| = 1$. 

\begin{lemma}
\label{LEMlocal}
Let $\mathbf{S}=(S_1,\ldots,  S_n)$ be a strategy profile 
in a matroid congestion game 
$(N, E, (\mathcal{S}_i)_{i\in N}, (c_e)_{e\in E}, \agg)$ 
with set-functional costs and complementarities. 
If $\mathbf{S}$ is not a pure Nash equilibrium, 
then 
there exists a player
$i\in N$ having an improving local move in $\mathbf{S}$. 
\end{lemma}

\begin{proof}
Since the strategy profile $\mathbf{S}$ is not a pure Nash equilibrium, 
there exists a player $i\in N$ having an improving move $S_i' \in \mathcal{S}_i$. 
Among the improving moves, 
choose $S_i'$ which minimizes $|S_i' \setminus S_i|$, 
and let $\mathbf{S}'$ denote the strategy profile $(\mathbf{S}_{-i},S_i')$. 
Since $S_i'$ is an improving move, 
it holds that 
\begin{align}
\label{EQSdashimproving}
\gamma_i(\mathbf{S}') < \gamma_i(\mathbf{S}).
\end{align}
We show $|S_i'\setminus S_i| =1$ 
to prove that $S_i'$ is an improving local move.

Suppose to the contrary that $|S_i' \setminus S_i| \ge 2$. 
It follows from 
Lemma \ref{LEMsimul} 
that 
there exist two resources 
$e^*\in S_i \setminus S_i'$ and  
$f^*\in S_i' \setminus S_i$ 
satisfying that 
$$T_i=(S_i \setminus \{e^*\}) \cup \{f^*\} \in \mathcal{S}_i \quad \mbox{and} \quad 
T_i'=(S_i' \setminus \{f^*\}) \cup \{e^*\} \in \mathcal{S}_i.$$ 
It then follows from $|S_i' \setminus S_i|\ge 2$ and the minimality of $|S_i' \setminus S_i|$ that 
neither $T_i$ nor $T_i'$ is an improving move, 
i.e., 
\begin{align}
\label{EQTnotimproving}
{}&{}\gamma_i(\mathbf{S}_{-i}, T_i) \ge \gamma_i(\mathbf{S}), \\
\label{EQTdashnotimproving}
{}&{}\gamma_i(\mathbf{S}_{-i}, T_i') \ge \gamma_i(\mathbf{S}). 
\end{align}

We have that 
\begin{align*}
\gamma_i(\mathbf{S}) 
{}&{}=g((c_e(N_e(\mathbf{S})))_{e\in S_i}), \\
\gamma_i(\mathbf{S}_{-i}, T_i) 
{}&{}= g((c_e(N_e(\mathbf{S}_{-i},T_i)))_{e\in T_i}) \\
{}&{}= g((c_e(N_e(\mathbf{S})))_{e\in S_i\setminus\{e^*\}}, c_{f^*}(N_{f^*}(\mathbf{S}_{-i},T_i))) \\{}&{}
= g((c_e(N_e(\mathbf{S})))_{e\in S_i\setminus\{e^*\}}, c_{f^*}(N_{f^*}(\mathbf{S})\cup\{i\})), 
\end{align*}
and hence it follows from \eqref{EQTnotimproving} that 
\begin{align}
\label{EQelef}
c_{e^*}(N_{e^*}(\mathbf{S}))\preceq_g c_{f^*}(N_{f^*}(\mathbf{S}) \cup \{i\}) = c_{f^*}
(N_{f^*}(\mathbf{S}')) .
\end{align}

Further, 
it follows from \eqref{EQSdashimproving} and \eqref{EQTdashnotimproving} that 
$\gamma_i(\mathbf{S}') <  \gamma_i(\mathbf{S}_{-i},T_i')$. 
We have that 
\begin{align*}
{}&{}\gamma_i(\mathbf{S}')=g( c_e(N_e(\mathbf{S}'))_{e\in S_i'} ),\\
{}&{}\gamma_i(\mathbf{S}_{-i},T_i')  = \gamma_i(\mathbf{S}_{-i}, (S_i' \setminus \{f^*\}) \cup \{e^*\}) 
                                = g( (c_e(N_e(\mathbf{S}')))_{e\in S_i'\setminus \{f^*\}},c_{e^*}(N_{e^*}(\mathbf{S})) ), 
\end{align*}
and hence 
\begin{align}
\label{EQflee}
c_{f^*}(N_{f^*}(\mathbf{S}')) \prec_g   c_{e^*}(N_{e^*}(\mathbf{S})).
\end{align}
We now have both \eqref{EQelef} and \eqref{EQflee}, a contradiction. 
\end{proof}

On the basis of Lemma \ref{LEMlocal}, 
we prove Theorem \ref{THMPNE} 
by showing that 
an improving local move 
strictly decreases a lexicographic order on the potential of the strategy profile. 
For this purpose, 
below we define the potential $\Phi(\mathbf{S})$ of a strategy profile $\mathbf{S}$
and 
a lexicographic order on the potentials. 

Let $|E|=m$. 
For a resource $e \in E$, 
define a two-dimensional vector $\phi_e(\mathbf{S}) \in \RR_+ \times \ZZ_+$ by 
$$\phi_e(\mathbf{S}) = (c_e(N_e(\mathbf{S})), n_e(\mathbf{S})).$$
\begin{definition}
For two resources $e,e' \in E$, 
we define 
$\phi_e(\mathbf{S}) \lex \phi_{e'}(\mathbf{S})$ 
if 
either
$$
c_e(\mathbf{S}) \prec_g c_{e'}(\mathbf{S})
$$
or 
$$
\mbox{$c_e(\mathbf{S}) \sim_g c_{e'}(\mathbf{S})$   and   $n_e(\mathbf{S}) \leq n_{e'}(\mathbf{S'})$}
$$
holds. 
Further, 
we define 
$\phi_e(\mathbf{S}) \lexeq \phi_{e'}(\mathbf{S})$ if 
$$\mbox{$\phi_e(\mathbf{S}) \lex \phi_{e'}(\mathbf{S})$ and  
$\phi_{e'}(\mathbf{S}) \lex \phi_{e}(\mathbf{S})$}$$
hold, 
and 
$\phi_{e}(\mathbf{S}) \lexneqq \phi_{e'}(\mathbf{S})$ if 
$$\mbox{$\phi_{e}(\mathbf{S})\lex \phi_{e'}(\mathbf{S})$ and 
$\phi_{e'}(\mathbf{S})\not\lex \phi_{e}(\mathbf{S})$}$$
hold. 
\end{definition}
Now 
the potential 
$\Phi(\mathbf{S})$ 
is a sequence $(\phi_e(\mathbf{S}))_{e \in E}$ of all resources in the lexicographically nonincreasing order. 
Namely, 
$$\mbox{$\Phi(\mathbf{S}) = (\phi_{e_1}(\mathbf{S}),\ldots, \phi_{e_m}(\mathbf{S}))$, 
where 
$\phi_{e_1}(\mathbf{S}) \succlex \cdots \succlex \phi_{e_m}(\mathbf{S})$}. $$

\begin{definition}[Lexicographic order on the potentials]
For two strategy profiles $\mathbf{S}$ and $\mathbf{S}'$, 
where 
$\Phi(\mathbf{S})= (\phi_{e_1}(\mathbf{S}),\ldots, \phi_{e_m}(\mathbf{S}))$
and 
$\Phi(\mathbf{S}')= (\phi_{e'_1}(\mathbf{S}),\ldots, \phi_{e'_m}(\mathbf{S}))$, 
we define $\Phi(\mathbf{S}) \lex \Phi(\mathbf{S'})$ 
if there exists an integer $k$ with $1 \le k \le m$ such that 
\begin{align*}
    &\mbox{$\phi_{e_j}(\mathbf{S}) \lexeq \phi_{e'_j}(\mathbf{S})$ for every $j < k$},\\
    &\phi_{e_k}(\mathbf{S}) \lexstrict\phi_{e'_k}(\mathbf{S}).    
\end{align*}
A strict relation $\Phi(\mathbf{S}) \lexstrict \Phi(\mathbf{S'})$ means that 
$\Phi(\mathbf{S}) \lex \Phi(\mathbf{S'})$ 
and 
$\Phi(\mathbf{S}')  \not\lex \Phi(\mathbf{S})$ hold. 
\end{definition}

We are now ready to prove Theorem \ref{THMPNE}. 

\begin{proof}[Proof of Theorem \ref{THMPNE}]

Let $\mathbf{S}=(S_1,\ldots,S_n)$ be a strategy profile, 
and suppose that it is not a pure Nash equilibrium. 
It then follows from Lemma \ref{LEMlocal} that
there exists a player $i \in N$ having a local improving move 
$S_i' = (S_i \setminus \{e^*\}) \cup \{f^*\}$.
Denote $\mathbf{S}'=(\mathbf{S}_{-i}, S_i')$. 
Recall that 
$\gamma_i(\mathbf{S}') < \gamma_i(\mathbf{S})$
holds. 

We analyze the lexicographic order of the potentials of  $\Phi(\mathbf{S})$ and $\Phi(\mathbf{S}')$. 
First, 
for each resource $e \in E \setminus \{e^*,f^*\}$, 
it holds that $N_e(\mathbf{S})=N_e(\mathbf{S}')$
and hence 
\begin{align}
\label{EQ5}
\phi_e(\mathbf{S}) = \phi_e(\mathbf{S}') \quad (e \in E \setminus \{e^*,f^*\}). 
\end{align}

We then investigate the two resources $e^*$ and $f^*$. 
First, 
it follows from $\gamma_i(\mathbf{S}') < \gamma_i(\mathbf{S})$ that 
\begin{align*}
g\left((c_e(N_e(\mathbf{S})))_{e \in S_i \setminus \{e^*\}}, c_{f^*}(N_{f^*}(\mathbf{S}')) \right) 
< 
g\left((c_e(N_e(\mathbf{S})))_{e \in S_i \setminus \{e^*\}}, c_{e^*}( N_{e^*}(\mathbf{S}) ) \right), 
\end{align*}
and hence 
\begin{align}
\label{EQ2}
c_{f^*}(N_{f^*}(\mathbf{S}')) \prec_g c_{e^*}(N_{e^*}(\mathbf{S})).
\end{align}
Further, 
it is straightforward to see that 
$i \in N_{e^*}(\mathbf{S})$ and 
$N_{e^*}(\mathbf{S}') = N_{e^*}(\mathbf{S})\setminus \{i\}$,
and thus 
\begin{align}
\label{EQ4}
{}&{}n_{e^*}(\mathbf{S}') < n_{e^*}(\mathbf{S}), \\
\label{EQ3}
{}&{}c_{e^*}(N_{e^*}(\mathbf{S}')) \preceq_g c_{e^*}(N_{e^*}(\mathbf{S})), 
\end{align}
where the latter relation \eqref{EQ3} follows from the monotonicity of 
the cost function $c_{e^*}$ (Definition \ref{DEFmonotonecost}). 

Now 
it follows from \eqref{EQ2} that 
\begin{align}
\phi_{f^*} (\mathbf{S}')
= ( c_{f^*}( N_{f^*}(\mathbf{S}') ),n_{f^*}(\mathbf{S'})  ) 
\lexstrict ( c_{e^*}(N_{e^*}(\mathbf{S}) ),n_{e^*}(\mathbf{S}) ) 
\label{EQ6}
= \phi_{e^*}(\mathbf{S}).
\end{align}
We also obtain from \eqref{EQ4} and \eqref{EQ3} that 
\begin{align}
\phi_{e^*}(\mathbf{S'}) 
= \left( c_{e^*}( N_{e^*} (\mathbf{S}') ),n_{e^*}(\mathbf{S'})  \right) 
\lexstrict \left( c_{e^*}( N_{e^*} (\mathbf{S}) ),n_{e^*}(\mathbf{S})  \right) 
\label{EQ7}
= \phi_{e^*}(\mathbf{S}).
\end{align}
From \eqref{EQ5}, \eqref{EQ6}, and \eqref{EQ7}, 
we derive 
that 
$
\Phi(\mathbf{S}') \lexstrict \Phi(\mathbf{S})$.

Therefore, an improving local move strictly decreases the 
lexicographic order of the potential $\Phi(\mathbf{S})$. 
On the basis of the finiteness of the game, 
after a finite number of improving local moves, 
a strategy profile in which no player has an improving local move is attained. 
It now follows from Lemma \ref{LEMlocal} that this strategy profile is a pure Nash equilibrium. 
\end{proof}

\section{Generalizations of Congestion Games with Mixed Costs}
\label{SECmixed-g}

In this section, 
we 
deal with generalizations of  congestion games with mixed costs, 
and 
provide some extensions of Theorem \ref{THMmixed}. 

\subsection{Player-Specific Congestion Games with Mixed Costs}

One model is a \emph{player-specific congestion game with mixed costs},
which is represented by a tuple 
\begin{align*}
(N, R, (\mathcal{S}_i)_{i \in N},  (\ell_{i,e})_{i\in N, e\in E}, (b_{i,e})_{i\in N,e\in E} , (\alpha_i)_{i \in N}).
\end{align*}
The difference from the model of congestion games with mixed costs in Section \ref{SECmixed} is 
that 
the latency and bottleneck cost functions are player-specific. 
Namely, 
each player $i\in N$ has her own latency cost function $\ell_{i,e}\colon \ZZ_+ \to \RR_+$ 
and bottleneck cost function $b_{i,e}\colon \ZZ_+ \to \RR_+$ for each resource $e\in E$. 
Hence, 
the cost $\gamma_i(\mathbf{S})$ imposed on a player $i \in N$ in a strategy profile $\mathbf{S}=(S_1,\ldots, S_n)$ is described as 
\begin{align*}
\gamma_i(\mathbf{S}) = 
\alpha_i \cdot \sum_{e \in S_i}\ell_{i,e}(n_e(\mathbf{S})) + (1 - \alpha_i) \cdot \max_{e \in S_i}\{ b_{i,e} (n_e(\mathbf{S})) \}. 
\end{align*}
We assume that the cost functions 
$\ell_{i,e}$ and $b_{i,e}$
are monotonically nondecreasing 
for each $i\in N$ and $e\in E$. 

The following theorem extends Theorem \ref{THMmixed}\ref{ENUsingleton}\ref{ENU01}. 
The proof below is a modest extension of that in \cite{FLS18}.

\begin{theorem}
A player-specific congestion game 
$G$
with mixed costs 
has a pure Nash equilibrium 
in the following two cases. 
\begin{enumerate}
\item
	\label{ENUsingleton_g}
	$G$ is a singleton game. 
\item
	\label{ENU01_g}
	$G$ is a matroid game and $\alpha_i \in \{0,1\}$ for each player $i\in N$. 
\end{enumerate}
\end{theorem}

\begin{proof}
We first investigate Case \ref{ENUsingleton_g}. 
Let $\mathbf{S} =(S_1,\ldots, S_n)$ be a strategy profile 
and let $i\in N$ be a player. 
Since $G$ is a singleton game, 
it holds that $S_i=\{e\}$ for some resource $e\in E$, 
and 
\begin{align*}
\gamma_i(\mathbf{S}) = 
\alpha_i \cdot \ell_{i,e}(n_e(\mathbf{S})) + (1 - \alpha_i)\cdot   b_{i,e} (n_e(\mathbf{S})) . 
\end{align*}
Now define a new player-specific cost function $c'_{i,e} \colon \ZZ_+ \to \RR_+$ by 
$$
c'_{i,e}(x) = \alpha_i \cdot\ell_{i,e}(x) + (1 - \alpha_i)\cdot   b_{i,e} (x) \quad (x \in \ZZ_+) 
$$
for each player $i \in N$ and each resource $e\in E$,   
and construct 
a player-specific singleton congestion game 
$$G'=(N, E, (\mathcal{S}_i)_{i \in N}, (c'_{i,e})_{i\in N, e\in E}),$$ 
in which the existence of a pure Nash equilibrium is proved in \cite{Mil96} (see also Theorem \ref{THMps}). 
Clearly, a pure Nash equilibrium in $G'$ is also a 
pure Nash equilibrium in $G$.

We next consider Case \ref{ENU01_g}.    
Let $\mathbf{S} =(S_1,\ldots, S_n)$ be a strategy profile. 
For a player $i\in N$, 
her cost $\gamma_i(\mathbf{S})$ is described as 
\begin{align*}
\gamma_i(\mathbf{S}) 
{}&{}=  \alpha_i\cdot \sum_{e \in S_i}\ell_{i,e}(n_e(\mathbf{S})) + (1 - \alpha_i)\cdot \max_{e \in S_i}\{ b_{i,e} (n_e(\mathbf{S})) \} \\
{}&{}=  
	\begin{cases}
	\displaystyle
	\sum_{e \in S_i}\ell_{i,e}(n_e(\mathbf{S})) 		& \mbox{if $\alpha_i = 1$},\\
	\displaystyle
	\max_{e \in S_i}\{ b_{i,e} (n_e(\mathbf{S})) \} 	& \mbox{if $\alpha_i = 0$}.
	\end{cases}
\end{align*}
Similarly to the above proof, 
define a new player-specific cost function $c''_{i,e}\colon \ZZ_+\to \RR_+$ by 
\begin{align*}
c''_{i,e}(x) =
	\begin{cases}
	\displaystyle
	\ell_{i,e}(x) 		& \mbox{if $\alpha_i = 1$},\\
	\displaystyle
	b_{i,e} (x)  	& \mbox{if $\alpha_i = 0$} 
	\end{cases}
	\quad (x\in \ZZ_+) 
\end{align*}
for each player $i\in N$ and each resource $e\in E$ 
to construct a player-specific matroid congestion game 
$$G''=(N, E, (\mathcal{S}_i)_{i \in N}, (c''_{i,e})_{i\in N, e\in E}),$$ 
which is guaranteed to possess a pure Nash equilibrium in \cite{ARV09} (Theorem \ref{THMps}). 
Since $G$ is a matroid game, 	
it follows from Lemma \ref{LEMgreedy} that 
a strategy $S_i\in \mathcal{S}_i$ minimizing 
$\sum_{e \in S_i}b_{i,e} (n_e(\mathbf{S}_{-i},S_i)) $ 
also minimizes 
$\max_{e \in S_i}\{ b_{i,e} (n_e(\mathbf{S}_{-i},S_i)) \}$.  
Now it is straightforward to see that 
a pure Nash equilibrium in $G''$ is also a 
pure Nash equilibrium in $G$. 
\end{proof}

\subsection{Weighted Congestion Games with Mixed Costs and Beyond}

Another model is a \emph{congestion game with mixed and set-functional costs}, 
which includes a \emph{weighted congestion game with mixed costs}. 
A congestion game with mixed and set-functional costs 
is represented by a tuple 
$$(N, E, (\mathcal{S}_i)_{i \in N}, (\ell_e)_{e\in E}, (b_e)_{e \in E}, (\alpha_i)_{i\in N}).$$
A new feature of this model is 
that 
the latency cost function $\ell_e\colon 2^N \to \RR_+$ and 
the bottleneck cost function $b_e\colon 2^N \to \RR_+$ are set functions for each resource $e\in E$. 
Again, 
these cost functions are assumed to be monotonically nondecreasing. 
Moreover, 
we generalize 
the definition of the monotone dependence 
of the cost functions $\ell_e$ and $b_e$ as 
the existence of a monotonically nondecreasing function $d\colon \RR_+ \to \RR_+$ satisfying that 
$$
b_e(N') = d(\ell_e(N')) \quad (N' \subseteq N)
$$
for each resource $e\in E$. 

For this model, 
we extend Theorem \ref{THMmixed}\ref{ENUmd} 
to obtain the following theorem. 
\begin{theorem}
A congestion game 
$G$
with mixed 
and set-functional costs 
has a pure Nash equilibrium 
if 
	$G$ is a matroid game and the cost functions have monotone dependence. 
\end{theorem}

\begin{proof}
Let $\mathbf{S}$ be a strategy profile of $G$ and 
$i \in N$ be a player. 
Since the cost functions have monotone dependence, 
it holds that 
\begin{align*}
\gamma_i(\mathbf{S}) 
{}&{}= \alpha_i \cdot \sum_{e\in S_i} \ell_{e}(N_{e}(\mathbf{S})) + 
(1-\alpha_i)\cdot \max_{e\in S_i} \{b_{e}(N_{e}(\mathbf{S}))\} \\
{}&{}= \alpha_i\cdot  \sum_{e\in S_i} \ell_{e}(N_{e}(\mathbf{S})) + (1-\alpha_i)\cdot 
\max_{e \in S_i}\{d( \ell_{e}(N_{e}(\mathbf{S})))\}
\end{align*}
for a monotonically nondecreasing function $d\colon \RR_+\to \RR_+$. 

Now construct a congestion game 
$G'=(N, E, (\mathcal{S}_i)_{i \in N}, (c_e)_{e\in E})$ with set-functional costs 
$c_e\colon 2^N\to \RR_+$ ($e\in E$)
defined by 
\begin{align*}
c_e(N') = \ell_e(N') \quad (N'\subseteq N). 
\end{align*}
It follows from Theorem \ref{THMmonotone} that 
$G'$ has  
a pure Nash equilibrium $\mathbf{S}=(S_1,\ldots,S_n)$.  

We have that $S_i$ minimizes $\sum_{e\in S_i} \ell_{e}(N_{e}(\mathbf{S}_{-i},S_i))$ 
among all the strategies in $\mathcal{S}_i$. 
Since $G$ is a matroid game, 
it follows from Lemma \ref{LEMgreedy} that 
$S_i $ also minimizes $\max_{e \in S_i} \{\ell_{e}(N_{e}(\mathbf{S}_{-i},S_i))\}$. 
From the monotonicity of $d$,
it follows that  
$$
\max_{e \in S_i} \{d(\ell_{e}(N_{e}(\mathbf{S})))\} = 
d\left(\max_{e \in S_i} \{\ell_{e}(N_{e}(\mathbf{S}))\}\right),
$$
and 
hence $S_i$ further minimizes $\max_{e \in S_i} \{d(\ell_{e}(N_{e}(\mathbf{S}_{-i},S_i)))\}$. 
Consequently, 
$S_i$ minimizes $\gamma_i(\mathbf{S}_{-i},S_i)$, 
and 
therefore 
we conclude that $\mathbf{S}$ is also a pure Nash equilibrium in $G$. 
\end{proof}

\section{Conclusion}

We have presented a model of congestion games 
with set-functional costs and complementarities, 
which commonly generalizes previous models 
such as 
weighted congestion games, 
congestion games with set-functional costs, 
congestion games with bottleneck costs, 
and 
congestion games with complementarities. 
Main technical contributions of this paper are 
a simplification of the monotonicity of the aggregation functions 
and 
a proof of the existence of pure Nash equilibria 
for matroid games under the simplified definition of the monotone aggregation functions. 
We have also extended the existence of pure Nash equilibria in 
congestion games with mixed costs 
to 
player-specific and weighted congestion games. 

A possible direction of future research is to analyze the algorithmic aspects and the efficiency of the pure Nash equilibria, 
such as 
the convergence time, 
the price of anarchy, 
and the price of stability. 
It is also of interest to investigate 
a relation between our model and 
congestion games with lexicographical improvement property \cite{HKM13}.

\bibliographystyle{abbrv}
\bibliography{refs}

\begin{thebibliography}{10}

\bibitem{ARV09}
H.~Ackermann, H.~R{\"{o}}glin, and B.~V{\"{o}}cking.
\newblock Pure {Nash} equilibria in player-specific and weighted congestion
  games.
\newblock {\em Theor. Comput. Sci.}, 410(17):1552--1563, 2009.

\bibitem{Aum59}
R.~J. Aumann.
\newblock Acceptable points in general cooperative $n$-person games.
\newblock In R.~D. Luce and A.~W. Tucker, editors, {\em Contribution to the
  Theory of Games IV}, pages 287--324. Princeton University Press, 1959.

\bibitem{BO07}
R.~Banner and A.~Orda.
\newblock Bottleneck routing games in communication networks.
\newblock {\em {IEEE} J. Sel. Areas Commun.}, 25(6):1173--1179, 2007.

\bibitem{BV23book}
V.~Bil{\`{o}} and C.~Vinci.
\newblock {\em Coping with Selfishness in Congestion Games---Analysis and
  Design via {LP} Duality}.
\newblock Monographs in Theoretical Computer Science. An {EATCS} Series.
  Springer, 2023.

\bibitem{Bru69}
R.~A. Brualdi.
\newblock Comments on bases in dependence structures.
\newblock {\em Bull. Austral. Math. Soc.}, 1:161--167, 1969.

\bibitem{DFRS15}
M.~Drees, M.~Feldotto, S.~Riechers, and A.~Skopalik.
\newblock On existence and properties of approximate pure {N}ash equilibria in
  bandwidth allocation games.
\newblock In M.~Hoefer, editor, {\em 8th International Symposium on Algorithmic
  Game Theory, {SAGT} 2015}, volume 9347 of {\em Lecture Notes in Computer
  Science}, pages 178--189. Springer, 2015.

\bibitem{DFRS19}
M.~Drees, M.~Feldotto, S.~Riechers, and A.~Skopalik.
\newblock Pure {N}ash equilibria in restricted budget games.
\newblock {\em J. Comb. Optim.}, 37(2):620--638, 2019.

\bibitem{FLS17}
M.~Feldotto, L.~Leder, and A.~Skopalik.
\newblock Congestion games with complementarities.
\newblock In D.~Fotakis, A.~Pagourtzis, and V.~T. Paschos, editors, {\em 10th
  International Conference on Algorithms and Complexity, {CIAC} 2017}, volume
  10236 of {\em Lecture Notes in Computer Science}, pages 222--233, 2017.

\bibitem{FLS18}
M.~Feldotto, L.~Leder, and A.~Skopalik.
\newblock Congestion games with mixed objectives.
\newblock {\em J. Comb. Optim.}, 36(4):1145--1167, 2018.

\bibitem{HHKS13}
T.~Harks, M.~Hoefer, M.~Klimm, and A.~Skopalik.
\newblock Computing pure {Nash} and strong equilibria in bottleneck congestion
  games.
\newblock {\em Math. Program.}, 141(1-2):193--215, 2013.

\bibitem{HKM13}
T.~Harks, M.~Klimm, and R.~H. M{\"{o}}hring.
\newblock Strong equilibria in games with the lexicographical improvement
  property.
\newblock {\em Int. J. Game Theory}, 42(2):461--482, 2013.

\bibitem{HKP18}
T.~Harks, M.~Klimm, and B.~Peis.
\newblock Sensitivity analysis for convex separable optimization over integral
  polymatroids.
\newblock {\em {SIAM} J. Optim.}, 28(3):2222--2245, 2018.

\bibitem{KT22}
F.~Kiyosue and K.~Takazawa.
\newblock A common generalization of budget games and congestion games.
\newblock In P.~Kanellopoulos, M.~Kyropoulou, and A.~A. Voudouris, editors,
  {\em 15th International Symposium on Algorithmic Game Theory, {SAGT} 2022},
  volume 13584 of {\em Lecture Notes in Computer Science}, pages 258--274.
  Springer, 2022.

\bibitem{Kuk07}
N.~S. Kukushkin.
\newblock Congestion games revisited.
\newblock {\em Int. J. Game Theory}, 36(1):57--83, 2007.

\bibitem{Kuk15}
N.~S. Kukushkin.
\newblock Rosenthal's potential and a discrete version of the {D}ebreu-{G}orman
  theorem.
\newblock {\em Autom. Remote. Control.}, 76(6):1101--1110, 2015.

\bibitem{Mil96}
I.~Milchtaich.
\newblock Congestion games with player-specific payoff functions.
\newblock {\em Game. Econ. Behav.}, 13:111--124, 1996.

\bibitem{MS96}
D.~Monderer and L.~S. Shapley.
\newblock Potential games.
\newblock {\em Games Econ. Behav.}, 14:124--143, 1996.

\bibitem{Mur00}
K.~Murota.
\newblock {\em Matrices and Matroids for Systems Analysis}.
\newblock Springer, 2010.

\bibitem{NRTV07}
N.~Nisan, T.~Roughgarden, {\'E}.~Tardos, and V.~V. Vazirani, editors.
\newblock {\em Algorithmic Game Theory}.
\newblock Cambridge University Press, 2007.

\bibitem{Rec89}
A.~Recski.
\newblock {\em Matroid Theory and Its Application in Electric Network Theory
  and in Statics}.
\newblock Springer, 1989.

\bibitem{Ros73a}
R.~W. Rosenthal.
\newblock A class of games possessing pure-strategy {Nash} equilibria.
\newblock {\em Int. J. Game Theory}, 2:65--67, 1973.

\bibitem{Rou16}
T.~Roughgarden.
\newblock {\em Twenty Lectures on Algorithmic Game Theory}.
\newblock Cambridge University Press, 2016.

\bibitem{Sch03}
A.~Schrijver.
\newblock {\em Combinatorial Optimization---Polyhedra and Efficiency}.
\newblock Springer, 2003.

\bibitem{Tak19}
K.~Takazawa.
\newblock Generalizations of weighted matroid congestion games: pure {N}ash
  equilibrium, sensitivity analysis, and discrete convex function.
\newblock {\em J. Comb. Optim.}, 38(4):1043--1065, 2019.

\bibitem{Tak24}
K.~Takazawa.
\newblock Pure nash equilibria in weighted congestion games with
  complementarities and beyond.
\newblock In N.~Alechina, V.~Dignum, M.~Dastani, and J.~Sichman, editors, {\em
  23rd International Conference on Autonomous Agents and Multiagent Systems,
  AAMAS 2024}. IFAAMAS, 2024.
\newblock to appear.

\end{thebibliography}

\end{document}